\documentclass[letterpaper, 10 pt, conference]{ieeeconf}  

\IEEEoverridecommandlockouts                              
\usepackage{cite}
\usepackage{amsmath,amssymb,amsfonts}
\usepackage{algorithmic}
\usepackage{graphicx}
\usepackage{textcomp}
\usepackage{setspace}
\usepackage{booktabs}

\usepackage{epsfig}
\usepackage{times} 
\usepackage{svg}
\usepackage[linesnumbered,ruled,vlined]{algorithm2e} 

\usepackage[nolist, nohyperlinks]{acronym}

\usepackage{bm}
\usepackage{breqn}
\usepackage{amsmath}
\usepackage{amssymb}
\usepackage{tabularx}

\usepackage{amsthm}

\newcommand{\Z}
{\mathbb{Z}}
\newcommand{\reals}{\mathbb{R}}
\newcommand{\R}{\reals}

\newcommand{\Acal}{\mathcal{A}}
\newcommand{\Bcal}{\mathcal{B}}
\newcommand{\Ccal}{\mathcal{C}}

\newcommand{\Ecal}{\mathcal{E}}

\newcommand{\Gcal}{\mathcal{G}}

\newcommand{\Kcal}{\mathcal{K}}
\newcommand{\Lcal}{\mathcal{L}}

\newcommand{\Ncal}{\mathcal{N}}

\newcommand{\Rcal}{\mathcal{R}}

\newcommand{\Ucal}{\mathcal{U}}
\newcommand{\Vcal}{\mathcal{V}}

\newcommand{\Zcal}{\mathcal{Z}}

\newcommand{\eqn}[1]{\begin{align} #1 \end{align}}
\newcommand{\eqnN}[1]{\begin{align*} #1 \end{align*}}
\newcommand{\seqn}[2][]{
\begin{subequations}
    #1
\begin{align} #2 \end{align}
\end{subequations}
}
\newcommand{\seqnone}[2][]{
\begin{equation} #1
\begin{aligned} #2 \end{aligned}
\end{equation}
}

\newcommand{\norm}[1]{\left\Vert #1 \right \Vert}

\theoremstyle{plain}
\newtheorem{theorem}{Theorem}
\newtheorem{corollary}{Corollary}
\newtheorem{lemma}{Lemma}
\newtheorem{problem}{Problem}
\newtheorem{definition}{Definition}
\newtheorem{prop}{Proposition}

\theoremstyle{definition}
\newtheorem{assumption}{Assumption}
\newtheorem{remark}{Remark}

\theoremstyle{remark}

\makeatletter
\let\NAT@parse\undefined
\makeatother
\usepackage{hyperref}
\hypersetup{
colorlinks, 
    linkcolor={red!50!black},
    citecolor={blue!50!black},
    urlcolor={blue!80!black}
}
\usepackage{cleveref}

\crefformat{problem}{Problem~#2#1#3}
\crefformat{assumption}{Assumption~#2#1#3}
\crefformat{prop}{Proposition~#2#1#3}

\usepackage[nolist,nohyperlinks]{acronym}
\acrodef{QP}[QP]{Quadratic Program}
\acrodef{QCQP}[QCQP]{Quadratically Constrained Quadratic Program}
\acrodef{MFCQ}[MFCQ]{Mangasarian-Fromovitz Constraint Qualification}
\acrodef{SSOC}[SSOC]{Strong Second Order Condition}
\acrodef{LICQ}[LICQ]{Linear Independence Constraint Qualification}
\acrodef{KKT}[KKT]{Karush–Kuhn–Tucker}

\title{\LARGE \bf Distributed Resilience-Aware Control in Multi-Robot Networks
}

\author{Haejoon Lee$^{1}$ and Dimitra Panagou$^{1,2}$
\thanks{*This work is supported by the Air Force Office of Scientific Research (AFOSR) under FA9550-23-1-0163}
\thanks{$^{1}$Department of Robotics,
        University of Michigan, Ann Arbor, MI, USA
        {\tt\small haejoonl@umich.edu}}%
\thanks{$^{2}$Department of Aerospace Engineering,
        University of Michigan, Ann Arbor, MI, USA
        {\tt\small dpanagou@umich.edu }}%
}

\begin{document}

\maketitle
\thispagestyle{empty}
\pagestyle{empty}

\begin{abstract}
Ensuring resilient consensus in multi-robot systems with misbehaving agents remains a challenge, as many existing network resilience properties are inherently combinatorial and globally defined. While previous works have proposed control laws to enhance or preserve resilience in multi-robot networks, they often assume a fixed topology with known resilience properties, or require global state knowledge. These assumptions may be impractical in physically-constrained environments, where safety and resilience requirements are conflicting, or when misbehaving agents share inaccurate state information. In this work, we propose a distributed control law that enables each robot to guarantee resilient consensus and safety during its navigation without fixed topologies using only locally available information. To this end, we establish a sufficient condition for resilient consensus in time-varying networks based on the degree of non-misbehaving or normal agents. Using this condition, we design a Control Barrier Function (CBF)-based controller that guarantees resilient consensus and collision avoidance without requiring estimates of global state and/or control actions of all other robots. Finally, we validate our method through simulations.
\end{abstract}

\section{Introduction}
The consensus problem in multi-agent systems finds numerous applications from control to optimization~\cite{zhu2019, tan2022}. However, the performance of consensus in general degrades in the presence of misbehaving agents that share faulty or incorrect information. Resilient consensus has therefore been extensively studied~\cite{LeBlanc13, wang2022, CDC2024, TAC2025, leblanc2012_continuous}. In particular, the \textit{Weighted Mean-Subsequence-Reduced} (W-MSR) algorithm was introduced in~\cite{LeBlanc13} to guarantee consensus among non-misbehaving or normal agents despite the presence of misbehaving agents.

The challenge of using W-MSR algorithm lies in its reliance on complex global topological network resilience properties, such as $r\textit{-robustness}$ and $(r,s)\textit{-robustness}$~\cite{LeBlanc13}. These properties are inherently combinatorial, making them difficult to compute online~\cite{zhang2015}. Moreover, in many multi-robot applications, links between robots are formed if their relative distances are within a threshold, which further complicates the practical implementation of these network properties in real-time operation.

To implement these network properties,~\cite{saldana2017,time_varying_strongly_usevitch20} enforce local inter-robot connectivity by maintaining predetermined topologies with known resilience levels. This approach eliminates the need for continuous resilience computation. However, its limitation is that the imposed connectivity constraints induce fixed topologies, making it impractical in dynamic or constrained environments where robots need flexible movements. There is earlier work considering the maintenance of network resilience without relying on fixed topologies~\cite{guerrero2020, saulnier2017, cavorsi2022, zhang2024, ICRA2025, saldana2017}, but these resilience-aware approaches rely on global state knowledge (i.e., states of all robots). To obtain global state knowledge, each robot needs to accurately estimate it based on the information from its neighbor robots, which becomes challenging as information being shared is unreliable in adversarial settings.

In this paper, we focus on the design of distributed controllers based on Control Barrier Functions (CBFs)~\cite{ames_cbf}, where each robot computes its control inputs using locally-available information only, so that it maintains a resilient structure and ensures resilient consensus without fixed topologies. The problem of developing CBF-based controllers for improving or maintaining resilience of multi-robot networks without relying on fixed network topologies has been studied in~\cite{guerrero2020,cavorsi2022,ICRA2025,zhang2024}. However, these approaches take a centralized CBF form, or assume that each robot can accurately estimate the global state knowledge and the control actions of all other robots at every time instant. These assumptions are unrealistic in complex and adversarial networks, and difficult to resolve due to the global nature of the resilience properties.

To address these problems, the main contribution of this paper is twofold. We develop a sufficient condition for resilient consensus in time-varying networks based only on the degree (i.e., the number of neighbors) of normal agents. While \cite{leblanc2012_continuous} establishes similar conditions based on degree of all agents in continuous-time consensus dynamics, our condition applies to discrete-time dynamics and specifically to the W-MSR algorithm. Then, leveraging this condition, we design a distributed controller that maintains each normal robot's degree above a threshold while preventing collisions. We prove that our controller guarantees resilience and safety using only local information under certain assumptions. In contrast to earlier work \cite{saulnier2017,guerrero2020, cavorsi2022, zhang2024}, our approach (a) focuses on local connectivity of normal agents rather than enforcing global resilience properties, (b) eliminates reliance on potentially unreliable global state estimation, and (c) explicitly accounts for inter-agent collisions, which is overlooked in previous CBF-based approaches~\cite{bhatia_decentralized_2024, de_carli_distributed_2024, capelli2020, cavorsi2022, guerrero2020, zhang2024}.

\emph{Organization}: 
In Section~\ref{sec:prelim}, we introduce the notation, preliminaries, and the problem statement. In Section~\ref{sec:minimum_degree} we present a sufficient condition for resilient consensus, and in Section~\ref{sec:controller} we design a distributed CBF-based controller for resilience. Simulation results evaluate our approach in Section~\ref{sec:sim}, while Section~\ref{sec:conc} concludes our work.

\section{Preliminaries and Problem Statement}
\label{sec:prelim}
We denote the cardinality of a set $\mathcal S$ as $|\mathcal S|$. We use $\R_+$, $\R_{\geq 0}$, and $\mathbb Z_{\geq 0}$ to denote sets of positive real, non-negative real, and non-negative integer numbers, respectively. We also use $\norm{\cdot}:\mathbb R^m\to \mathbb R_{\geq 0}$ to represent the Euclidean norm.

\subsection{Multi-Robot Network System}

Consider a multi-robot system of $n$ robots. Robot $i\in\{1,\dots,n\}$ has a physical state $x_i(t) \in \mathbb R^m$ and control input $u_i(t) \in \Ucal_i\subset \R^m$, denoting position and velocity at time $t$, respectively. Each robot $i$ has single integrator dynamics as
\begin{equation}
    \dot {x}_i(t) = u_i(t).
    \label{eq:dynamics}
\end{equation}
We model robots as circular disks with radii $\frac {\Delta_d}{2}$, where $\Delta_d\in \R_+$. We also define the state and control input of the entire system as $\mathbf x(t)=\begin{bmatrix}
    x_1(t)^T \cdots x_n(t)^T
\end{bmatrix}^T \in \mathbb R^{nm}$ and $\mathbf u(t)=\begin{bmatrix}
    u_1(t)^T \cdots u_n(t)^T
\end{bmatrix}^T \in \Ucal = \Ucal_1\times \dotsm\times \Ucal_n \subset \R^{nm}$. Then the system dynamics is 
\eqn{\label{eq:combined_dynamics}\dot {\mathbf{x}}(t) = \mathbf{u}(t).} 
The distance between robots $i$ and $j$ is denoted as $\Delta_{ij}(\mathbf x(t))=\norm{x_i(t)-x_j(t)}$. We model the multi-robot network as a time-varying simple, undirected graph $\mathcal G(t)=(\mathcal V, \mathcal E(t))$, where the vertex set $\mathcal V=\{1,\dots, n\}$ includes robots, and the edge set $\mathcal E(t)=\{(i,j) \mid \Delta_{ij}(\mathbf x(t))\leq R\}$ defines connections (links) between two robots within a communication radius $R \in \R_+$. The set of neighbors of robot $i\in \Vcal$ at time $t$ is denoted as $\mathcal N_i(t) =\{j \in \mathcal V \mid (i,j) \in \mathcal E(t)\}$. We also denote $\Bcal_i(t)=\Ncal_i(t)\cup \{i\}$. Each robot can observe states of its own and of its neighbors. The degree of robot $i$ at time $t$ is denoted as $\delta_i(\Gcal(t))=|\Ncal_i(t)|$. Then, the minimum degree of $\Gcal(t)$ is denoted as $\delta_{\min}(\Gcal(t))=\min_{i\in \Vcal} \delta_i(\Gcal(t))$. 

We smoothly approximate the network connections through the weighted adjacency matrix $\mathbf A(\mathbf x(t))$, whose entries are
\begin{equation}a_{ij}(\mathbf x(t)) = \begin{cases}
     \frac {q_1} {1+e^{-q_{2}(R^2-\Delta_{ij}(\mathbf x(t))^2)^2}} - \frac {q_1} 2 & \Delta_{ij}(\mathbf x(t)) \leq R, \\
     0 & \text{otherwise},
\end{cases}    \label{eq:adjacency}
\end{equation}
where $q_1,q_2\in \mathbb R_+$. For simplicity, we drop $t$ from $\mathbf x(t)$, $x_i(t)$, $\mathbf u(t)$, and $u_i(t)$ when the context is clear. The function \eqref{eq:adjacency} is continuously differentiable with respect to $\mathbf x$. 
At discrete times $t_{l} = l\tau_1$ for $l\in \mathbb Z_{\geq0}$, where $\tau_1\in \mathbb R_+$ is an update interval, robot $i$ shares its 
connectivity level $c_i(t_l)\in \R_{\geq 0}$ with its neighbors. Robot $i$ maintains and updates its connectivity level as
\eqn{\label{eq:connectivity}c_i(t) = \sum_{j\in \Ncal_i(t_l)}a_{ij}(\mathbf x(t_l)), \quad \forall t \in [t_l, t_{l+1}).}
Each neighbor who receives it also maintains the value of~\eqref{eq:connectivity} for $t\in [t_l, t_{l+1})$ until the updated value is shared at $t_{l+1}$.
To simplify our analysis on our main results, we assume $\tau_1$ is sufficiently small (i.e., $\tau_1 \to 0$) such that 
\eqn{\label{eq:continuous_a}c_i(t)=\sum_{j\in \Ncal_i(t)} a_{ij}(\mathbf x(t))}
for all $t$; that is, each agent shares with neighbors in real time.


\subsection{Weighted Mean Subsequence Reduced Algorithm}
At discrete times $t_{p} =p\tau_2 $ for $p\in \mathbb Z_{\geq0}$ where $\tau_2\in \mathbb R_+$ (not necessarily the same as $\tau_1$) is the update interval, robot $i$ also shares its consensus state $y_i(t_p) \in \mathbb R$ with its neighbors. Robot $i$ has $y_i(t)=y_i(t_p)$ $\forall t \in [t_p, t_{p+1})$ and an update protocol:
\begin{equation}
y_i(t_{p+1})=\sum_{j\in \Bcal_i(t_p)} w_{ij}(t_p)y_j(t_p),
    \label{eq:linear}
\end{equation}
where $w_{ij}(t_p)$ is the weight assigned to $y_j(t_p)$ by agent $i$. We assume $\exists \beta \in(0,1)$ such that $\forall i \in \mathcal V$ $\forall p \in \mathbb Z_{\geq 0}$,
\begin{itemize}
    \item $w_{ij}(t_p)\geq \beta$ if $j\in \Bcal_i(t_p)$, or $w_{ij}(t_p)=0$ otherwise,
    \item $\sum_{j=1}^n w_{ij}(t_p)=1$.
\end{itemize}
Agents are guaranteed to achieve consensus on $y_i(t)$ through the protocol \eqref{eq:linear} if there exists a uniformly bounded sequence of contiguous time intervals where the union of $\Gcal(t)$ in each interval is connected \cite[Theorem 3.10]{discrete_time_varying_consensus_sufficient}. However, when $\Gcal$ has misbehaving and non-misbehaving (normal) agents, the consensus guarantee no longer holds. In this paper, we consider one specific type of misbehaving agents:
\begin{definition}[\textbf{malicious robot}]
    \label{misbeh}
    A robot $i\in \Vcal$ is \textbf{malicious} if it sends $c_i(T_1)$ and $y_i(T_2)$ to all its neighbors at times $T_1$ and $T_2$, respectively, but does not follow the nominal update protocols~\eqref{eq:continuous_a} and/or~\eqref{eq:linear} at some times $T_1$ and/or $T_2$.
\end{definition}

\noindent We define sets of normal and malicious agents as $\Lcal \subseteq \Vcal$ and $\Acal \subseteq \Vcal$, respectively, such that $\Lcal \cap \Acal =\emptyset$ and $\Lcal \cup \Acal = \Vcal$. Various threat models describe the number of malicious robots in a network \cite{LeBlanc13}, and we consider the following model:
\begin{definition}[$\mathbf F$-\textbf{total}]
    \label{ftotal}
    A set $\mathcal S \subset \mathcal V$ is $\mathbf F$-\textbf{total} if it contains at most $F$ nodes in the graph (i.e. $|\mathcal S|\leq F)$.
\end{definition}

Then resilient consensus is defined as \cite{LeBlanc13}:
\begin{definition}[\textbf{resilient consensus}]
Under an $F$-total attack, the normal nodes are said to achieve \textbf{resilient consensus} if both of the conditions hold for any initial consensus states:
\label{def:consensus}
\begin{itemize}

    \item $\exists M\in \R$ such that $\lim_{p\to \infty} y_i(t_p)=M$ $\forall i\in \Lcal$.
    
    \item $y_i(t_p) \in [\underline{z}(0),\bar{z}(0)]$, $\forall i \in \Lcal$, $\forall p \in \Z_{\geq 0}$, where $\underline{z}(0)=\min_{i \in \Lcal} y_i(0)$ and $\bar z(0)=\max_{i \in \Lcal} y_i(0)$.
\end{itemize}
\end{definition}

\begin{remark}
Robots share two different values, $c_i(t)$ and $y_i(t)$, at different update intervals. Notice that while malicious robots may manipulate one or both values by not following the nominal protocols~\eqref{eq:continuous_a} and/or~\eqref{eq:linear}, normal robots focus solely on achieving resilient consensus on $y_i(t)$. 
\end{remark}

In~\cite{LeBlanc13}, the W-MSR (\textit{Weighted-Mean Subsequent Reduced}) algorithm was introduced to guarantees resilient consensus. A normal robot $i\in \Lcal$ performs W-MSR algorithm with parameter $F$ at every time step $t_p$ as follows:
\begin{enumerate}
     \item Robot $i$ collects values $y_j(t_p)$ from its neighbors $j\in \Ncal_i(t_p)$ and sorts them.
    \item If there are at least $F$ strictly larger values than its own $y_i(t_p)$, it discards the largest $F$; otherwise, it removes all larger values. Similarly, it does the same for values less than $y_i(t_p)$. 
    \item Let $\Rcal_i(t_p)$ be the set of robots whose values are discarded in Step 2. Then, it applies
    \eqn{y_i(t_{p+1})=\sum_{j\in \Bcal_i(t_p)\setminus \Rcal_i(t_p)} w_{ij}(t_p)y_j(t_p),}
    where $w_{ij}(t_p)$ satisfies the same conditions stated for the nominal consensus update protocol~\eqref{eq:linear}.  
\end{enumerate}
W-MSR algorithm ensures that normal robots can achieve resilient consensus despite the presence of up to $F$ malicious agents, under certain topological properties of the communication graph. We review these relevant properties below:
\begin{definition}[$\mathbf {(r,s)}$-\textbf{reachable}]
    \label{reachability}
  Given a graph $\mathcal G(t)=(\mathcal V, \mathcal E(t))$ and
a nonempty subset $\mathcal S \subset \mathcal V$, $\mathcal S$ is $\mathbf {(r,s)}$-\textbf{reachable} at time $t$ if there are at least $s \in \mathbb Z_{\geq 0}$ nodes in $\mathcal S$ with at least $r \in \mathbb Z_{\geq 0}$ neighbors outside of $\mathcal S$ at $t$; i.e., given $\mathcal X^r_{\mathcal S}(t) = \{i \in \mathcal S : |\mathcal N_i(t)\setminus \mathcal S| \geq r\}$, $|\mathcal X^r_{\mathcal S}(t)|\geq s$.

\end{definition}

\begin{definition}[$\mathbf {(r,s)}$-\textbf{robust}]
    \label{def:robustness}
    A graph $\mathcal G(t) = (\mathcal V,\Ecal(t))$ is $\mathbf {(r,s)}$-\textbf{robust} at time $t$ if for every
pair of nonempty, disjoint subsets $\mathcal S_1,\mathcal S_2 \subset \mathcal V$, at least one
of the following holds at $t$ (recall $\mathcal X^r_{\mathcal S_k}(t) = \{i \in \mathcal S_k : |\mathcal N_i(t) \setminus \mathcal S_k| \geq r\}$ for $k\in\{1,2\}$):
\begin{itemize}
    \item $|\mathcal X^r_{\mathcal S_1}(t)|= |\mathcal S_1|$;
    \item $|\mathcal X^r_{\mathcal S_2}(t)|= |\mathcal S_2|$;
    \item $|\mathcal X^r_{\mathcal S_1}(t)|$ + $|\mathcal X^r_{\mathcal S_2}(t)|\geq s$.
\end{itemize}
\end{definition}
A network being $(F+1,F+1)$-robust at all times is a sufficient condition for its normal robots to achieve resilient consensus through the W-MSR algorithm under an $F$-total malicious attack~\cite[Corollary 1]{LeBlanc13}. Note for a graph of $n$ nodes, $F\leq\lfloor \frac{n-1}{2}\rfloor$, which is a direct application of~\cite[Lemma 4]{LeBlanc13}. Computing $(r,s)$-robustness of a network is challenging due to its combinatorial nature~\cite{zhang2015}, making its integration into multi-robot systems difficult. We build on results from this notion of resilience to present an alternative sufficient condition based on the minimum degree of normal agents, offering a more computationally efficient approach.

\subsection{Control Barrier Function}
In this work, we leverage Control Barrier Functions (CBFs)~\cite{ames_cbf} to ensure resilient consensus despite an $F$-total attack. Note we are considering time-invariant CBFs. Let a continuously differentiable function $h:\mathcal D \subset \mathbb R^{nm} \to \mathbb R$ define a safe set 
$\Ccal=\{\mathbf x \in \R^{nm} \mid h(\mathbf x)\geq 0 \}$ and its boundary $\partial \Ccal=\{\mathbf x \in \R^{nm} \mid h(\mathbf x)=0\}$. We say $h$ is a CBF for the system~\eqref{eq:combined_dynamics} if there exists an extended class $\Kcal_{\infty}$ function $\alpha:\R \to \R$ such that 
\eqn{\sup_{\mathbf u \in \Ucal} \Big[\frac {\partial h}{\partial \mathbf x} \dot {\mathbf x} +\alpha(h(\mathbf x))\Big]\geq 0, \quad \forall \mathbf x \in \Ccal.\label{eq:cbf}}
If $h$ is a CBF on a given set $\Ccal$ for given system dynamics and $\frac {\partial h}{\partial \mathbf x}(\mathbf x)\neq 0$ $\forall \mathbf x \in \partial \mathcal C$, then any locally Lipschitz continuous controller $\mathbf u^*(\mathbf x) \in K_{\textup{cbf}}(\mathbf x)=\{\mathbf u \in \Ucal \mid \frac {\partial h}{\partial \mathbf x} \dot {\mathbf x} \geq -\alpha(h(\mathbf x))\}$ renders the set $\Ccal$ forward invariant~\cite[Theorem~2]{ames_cbf}. 

\subsection{Problem Statement}
Let $\Gcal(t)$ be a multi-robot network formed by $n$ robots' physical states $\mathbf x$ at time $t\geq t_0$, where $t_0 \in \R_{\geq 0}$ is the initial time. Each robot $i\in \Vcal$ also has its own its connectivity level $c_i(t)$ and consensus state $y_i(t)$, which it shares with its neighbors. Robots are tasked with achieving consensus on $y_i(t)$ while navigating using a desired controller $u_{i,\text{des}}:\mathbb R^m\rightarrow \mathcal U_i$. Our objective is to design distributed control strategies that enable normal robots to achieve resilient consensus without fixed communication topologies, while also minimally deviating from their desired controllers. Ensuring resilient consensus in this system requires control laws for both the consensus states $y_i(t)$ and the physical states $x_i(t)$ of robots, which we address separately. This paper focuses on cases where malicious robots disrupt consensus through information only, that is:
\begin{assumption}
\label{assum:attack}
   Malicious robots follow the same control laws for their physical states as normal robots and do not engage in physical disruptions (e.g., collisions).
\end{assumption}
\noindent This assumption is reasonable from the perspective that, in principle, normal robots may have detection mechanisms to detect anomalous behaviors; hence it is reasonable to assume that malicious robots aim to mimic the behavior of the normal agents' physical states to remain undetected. The consideration of both physical and informational attacks is future work.

To update consensus states, robots adopt W-MSR algorithm~\cite{LeBlanc13}, which ensures resilient consensus given that the network topology is $(F+1,F+1)$-robust $\forall t\geq t_0$.
The main challenge of using W-MSR algorithm lies in the fact that $(r,s)$-robustness is combinatorial and a function of global network states (i.e., the states of all robots). Existing approaches for maintaining these properties typically require obtaining global state information through inter-agent communication. However, such communication becomes unreliable in the presence of malicious agents. Thus, we present an alternative sufficient condition that is locally controllable.
\begin{problem}
\label{prob:prob1}
Given a network $\Gcal(t)=(\Vcal, \Ecal(t))$ under an $F$-total attack at time $t$, let $\delta_{\Lcal}(\Gcal(t))=\min_{i\in \Lcal}(\delta_i(\Gcal(t)))$ be the minimum degree of normal agents at time $t$. We aim to find the lower bound for $\delta_{\Lcal}(\Gcal(t))$ $\forall t \geq t_0$ in which resilient consensus is guaranteed.
\end{problem}

Addressing~\Cref{prob:prob1} reduces the challenge of ensuring resilient consensus with $F$ malicious agents to ensuring each normal robot maintains its number of neighbors above a threshold $F'\in \Z_+$. The value $F'$, which depends on $F$, is provided in~\Cref{thm:rs_robust}. Given~\Cref{assum:attack}, we solve:
\begin{problem}
\label{prob:prob2}
 Let $\Gcal(t_0)=(\Vcal, \Ecal(t_0))$ be a network under an $F$-total attack, where $\delta_{\Lcal}(\Gcal(t_0))\geq F'$. Let $u_{i,\textup{des}}$ be the desired controller for agent $i\in \Vcal$. We design a distributed control strategy for all robots that requires only locally available information to ensure $\delta_{\Lcal}(\Gcal(t))\geq F'$ without collisions, while minimally deviating from $u_{i,\textup{des}}, \ \forall t\geq t_0$. 
\end{problem}
We propose a control strategy (given in~\eqref{seq:qp}) such that when all agents implement it using local information, they together ensure that normal robots maintain at least $F'$ neighbors without collisions under some assumptions, thereby achieving resilient consensus.

\section{Degree and Resilient Consensus}
\label{sec:minimum_degree}
Consider a graph $\Gcal(t)=(\Vcal, \Ecal(t))$ formed by $n$ robots' states $\mathbf x$ at time $t \in \R_{\geq 0}$. In this section, we solve~\Cref{prob:prob1} by showing that having $\delta_{\Lcal}(\Gcal(t))=\min_{i\in \Lcal}(\delta_i(\Gcal(t)))$ above some given threshold $F' \in \Z_+$ $\forall t\geq t_0$ guarantees resilient consensus. The main advantage of using $\delta_{\Lcal}(\Gcal(t))$ is that normal robots can directly control and observe their influence on it as they navigate at time $t$. The condition is given in~\Cref{thm:rs_robust}, but before that, we provide a useful result from~\cite[Property 5.23]{leblanc-thesis} that characterizes the $(r,s)$-robustness of $\Gcal(t)$ in terms of its minimum degree. 
\begin{lemma}
    Let $\mathcal G(t)=(\mathcal V, \mathcal E(t))$ be a graph with the minimum degree of~$\delta_{\min}(\mathcal G(t))$ at time $t$. Then, if $\delta_{\min}(\Gcal(t))\geq \lfloor \frac n 2 \rfloor-1$, $\mathcal G(t)$ is $(\delta_{\min}(\mathcal G(t))-\lfloor{\frac n 2}\rfloor+1,s)$-robust, for all $1\leq s\leq n$~\cite{leblanc-thesis}.
    \label{lem:rs_robust}
\end{lemma}

\Cref{lem:rs_robust} provides a sufficient condition for $(r,s)$-robustness of a graph $\Gcal(t)$ using its minimum degree $\delta_{\min}(\Gcal(t))$ at time $t$. 
In the context of~\Cref{lem:rs_robust}, the guarantees of $(r,s)$-robustness and hence resilient consensus depend on the local connectivity of all nodes, both normal and malicious agents. We relax this dependence to derive a sufficient condition for resilient consensus based solely on the connectivity of the \textit{normal} agents. 


\begin{prop}
    Let $\mathcal G(t)=(\mathcal V, \mathcal E(t))$ be a time-varying graph under an $F$-total malicious attack. If $\delta_{\Lcal}(\Gcal(t))\geq F'=F+\lfloor \frac n 2 \rfloor$ $\forall t\geq t_0$, normal agents are guaranteed to achieve resilient consensus through the W-MSR algorithm with parameter $F$.
\label{thm:rs_robust}    
\end{prop}

\begin{proof} Let $\Gcal_{\Lcal}(t) = (\Lcal, \Ecal_{\Lcal}(t))$ be a subgraph of $\Gcal(t)$ induced by the set of normal agents $\Lcal$ with $n':=|\Lcal|=n-F\geq F+1$. First, we analyze $\Gcal_\Lcal(t)$'s connectivity. All $n'$ normal agents must have at least $\lfloor \frac {n} 2 \rfloor = \lfloor \frac {n'+F} 2 \rfloor$ normal neighbors. Then by~\cite[3.8]{Fiedler1973}, $\Gcal_{\Lcal}(t)$ is connected $\forall t\geq t_0$.

Now, consider two time-varying graphs of $n$ nodes: (a) $\Gcal_1(t)=(\Vcal_1,\Ecal_1(t))$ such that $\delta_{\min}(\Gcal_1(t))\geq F'$ $\forall t\geq t_0$ and (b) $\Gcal_2(t)=(\Vcal_2,\Ecal_2(t))$ such that $\delta_{\Lcal}(\Gcal(t))\geq F'$ $\forall t\geq t_0$ and $\exists T\geq t_0$ such that $\delta_{\min}(\Gcal_2(T))<F'$. Since $\delta_{\min}(\Gcal_1(t))\geq F'$ $\forall t \geq t_0$, all agents (both normal and malicious agents) in $\Gcal_1(t)$ have at least $F'$ neighbors $\forall t \geq t_0$. On the contrary, all normal agents in $\Gcal_2(t)$ have at least $F'$ neighbors, whereas at least one malicious agent $i\in \Acal$ in $\Gcal_2(t)$ has less than $F'$ neighbors at $t=T$. By~\Cref{lem:rs_robust}, $\Gcal_1(t)$ is $(F+1,F+1)$-robust for all $t$, which is a sufficient condition for the achievement of resilient consensus under $F$-total malicious attacks~\cite[Corollary~1]{LeBlanc13}. Thus in $\Gcal_1(t)$, normal agents have enough redundancy in information from normal agents by having $F'$ neighbors $\forall t \geq t_0$ to achieve resilient consensus through the W-MSR algorithm. Now, assume to the contrary that resilient consensus is not achieved among normal agents in $\Gcal_2(t)$. Note that resilient consensus is achieved when all normal agents achieve consensus, regardless of the consensus states of malicious agents. Furthermore, in $\Gcal_2(t)$, normal agents have $F'$ or more neighbors $\forall t\geq t_0$. This implies that they also have enough redundancy in information from normal agents $\forall t\geq t_0$ to achieve resilient consensus through W-MSR algorithm, which is a contradiction. 
\end{proof}

Note that the same bound for the entire graph (including both normal and malicious agents) appears in~\cite[Corollary 1]{leblanc2012_continuous}, but in the context of continuous-time consensus dynamics and a different algorithm. 
Based on~\Cref{thm:rs_robust}, to guarantee resilient consensus under $F$-total malicious attacks, it suffices that robots navigate the environment such that \textit{normal robots} have at least $F'=F+\lfloor{\frac n 2}\rfloor$ neighbors $\forall t\geq t_0$. Since the numbers of total robots $n$ and malicious robots $F$ are fixed, each robot $i \in \Vcal$ can compute $F'$ locally. Now we design a control law for the physical states of robots so that normal robots always maintain $F'$ neighbors to achieve resilient consensus without collisions.

\section{Distributed Resilience-Aware Controller}

\label{sec:controller}
In this section, we design a distributed controller for robots' physical states that enables robots to navigate the environment with guarantees of resilient consensus and physical safety. Consider a multi-robot network $\Gcal(t)$ formed by the physical states $\mathbf x(t)$ of $n$ robots at time $t\geq t_0$. By~\Cref{thm:rs_robust}, $\delta_{\Lcal}(\mathcal G(t))\geq F' = F + \lfloor{\frac n 2}\rfloor$ $\forall t\geq t_0$ guarantees resilient consensus through the W-MSR algorithm. Thus, we design a controller that ensures that all normal robots maintain $F'$ or more neighbors without collisions $\forall t\geq t_0$, thereby addressing~\Cref{prob:prob2}.

There exist many works on decentralized control for connectivity maintenance \cite{ de_carli_distributed_2024,bhatia_decentralized_2024, kconnectivity2019}. However, these methods either (a) require robots to obtain states of all robots through inter-agent communication, which becomes unreliable in the presence of malicious agents, or (b) enforce fixed topologies, which may not always be feasible in certain settings or environments. Thus, we design a controller that (a) uses only local information while being robust to malicious sharing and (b) preserves a flexible network formation.

How we construct the controller (given in~\eqref{seq:qp}) is as follows: first, we define constraint functions to enforce degree maintenance and collision avoidance. These functions are then composed into one centralized candidate CBF~\eqref{eq:composition} whose corresponding CBF constraint~\eqref{eq:cbf_constraint} involves global state knowledge and control actions of all robots. Next, we decompose the constraint to create $n$ separate \textit{distributed} constraints~\eqref{eq:split} - one constraint for each robot - that only requires local information. Finally, we construct the controller with these distributed constraints, proving that it guarantees degree maintenance and safety under some assumptions.

\subsection{Constraint with Global State Knowledge}
\subsubsection{Constraint for Degree Maintenance} We define the function that encodes maintenance of $F'$ neighbors for robot $i\in \Vcal$:
\begin{equation}
    h_i(\mathbf x) = \sum_{j\in \Ncal_i(t)} a_{ij}(\mathbf x)-F'.
    \label{eq:h_i}
\end{equation} 
Since $a_{ij}(\mathbf x)$~\eqref{eq:adjacency} is continuously differentiable, $h_i(\mathbf x)$ is also continuously differentiable. Note that although we define $n$ CBFs so that all robots (both normal and malicious) have the same controller by~\Cref{assum:attack}, as will be shown in~\Cref{remark:diff}, the guarantee of maintaining $F'$ neighbors may apply only to normal agents. 

\begin{lemma}
Let $a_{ij}(\mathbf x)$ be computed from~\eqref{eq:adjacency} with $q_1=2+\epsilon$, where $0<\epsilon< \frac {1} {n-1}$, and $q_2\in \R_+$ at time $t$. Then,  $h_i(\mathbf x)\geq 0$ only if $\delta_i(\Gcal(t))\geq F'$ at time $t$.
\label{lem:connectivity}
\end{lemma}
\begin{proof}
    With $q_1=2+\epsilon$, where $0<\epsilon< \frac {1} {n-1}$, and $q_2\in \R_+$ at time $t$, $0\leq a_{ij}(\mathbf x) <1+\frac {1} {n-1}$. Then, given $K= \delta_i(\Gcal(t))\leq n-1$,  $\sum_{j \in \Ncal_i(t)} a_{ij}(\mathbf x)<K+ \frac K {n-1}\leq K+1$ $\forall t \geq t_0$. Thus, $\sum_{j\in \Ncal_i(t)} a_{ij}(\mathbf x)\geq F'$ only if $\delta_i(\Gcal(t))\geq F'$.
\end{proof}
With~\Cref{lem:connectivity}, $h_i(\mathbf x)\geq 0$ only if robot $i$ has at least $F'$ neighbors. 
Thus, if $\mathbf x \in \mathcal C_i=\{\mathbf x \in \R^{nm}\mid h_i(\mathbf x)\geq0\}$ at time $t$, all robots are positioned in such a way that $\delta_i(\mathcal G(t)) \geq F'$ at $t$. Since we want $\delta_{\Lcal}(\Gcal(t))\geq F'$ $\forall t \geq t_0$, the safe set of interest is $\Ccal =\cap_{i \in \Lcal} \Ccal_i$. Also,~\eqref{eq:h_i} does not specify which neighbors each robot should connect to, allowing robots the flexibility to form links dynamically. This enables the overall network to be time-varying without fixed topologies.


\subsubsection{Constraint for Inter-Agent Collision Avoidance} 
For realistic deployment, we also incorporate collision avoidance between agents to ensure safe navigation.
In this work, we construct candidate CBFs similar to those in \cite{cavorsi2022, capelli2020} to enforce collision avoidance between agents $i$ and $j$:
\eqn{h^{\textup{col}}_{ij}(\mathbf x) =  \Delta_{ij}(\mathbf x)^2 - \Delta_d^2 \label{eq:collision}.}
When $h_{ij}^{\text{col}}(\mathbf x)\geq 0$, this implies robots $i$ and $j$ are at least $\Delta_d$ away from each other, ensured to be free of collision.

\subsubsection{Composition of Multiple Constraints}
We compose all constraint functions into one candidate CBF. Composition methods are studied in \cite{black2023, glotfelter2017}, and in this paper we use the form inspired by \cite{black2023}. Then, we define:
\begin{equation}
   \phi(\mathbf x,\mathbf w) = 1 - \sum_{i\in \Vcal} E_i(\mathbf x)-\sum_{(i,j) \in \Ecal(t)} E^c_{ij}(\mathbf x),
\label{eq:composition}
\end{equation}
where $E_i(\mathbf x)=e^{-w_r {h}_i(\mathbf x)}$ and $E^c_{ij}(\mathbf x) = e^{-w_ch_{ij}^{\text{col}}(\mathbf x)}$ with a parameter $\mathbf w = [w_r, w_c] \in \R^2_{+}$. Its structure implies that $\phi(\mathbf x,\mathbf w)\geq 0$ only if $ {h}_i(\mathbf x)>0$ and $ h^{\textup{col}}_{ij}(\mathbf x)>0$ $\forall j \in \Ncal_i(t)$ and $\forall i \in \Vcal$. When $\mathbf x \in \mathcal Z(\mathbf w)=\{\mathbf x \in \mathbb R^{nm} \mid \phi(\mathbf x,\mathbf w)\geq0\} \subset \Ccal$, all robots have at least $F'$ neighbors (by~\Cref{lem:connectivity}) and stay at least $\Delta_d$ apart. Also let $\partial \Zcal(\mathbf w)=\{\mathbf x \in \R^{nm} \mid \phi(\mathbf x, \mathbf w)=0\}$. Then, \eqref{eq:composition} is a CBF for system~\eqref{eq:combined_dynamics} if there exists an extended class $\Kcal_{\infty}$ function $\alpha:\R \to \R$ such that
\eqn{
\label{eq:cbf_constraint}
 \sup_{\mathbf u \in \Ucal} \Bigg[ \sum_{i \in \Vcal} &
\bigg(w_r\sum_{k\in \Bcal_i(t)} E_k(\mathbf x)\frac {\partial h_k}{\partial x_i}  \\ &+ w_c\sum_{j\in \Ncal_i(t)} E^c_{ij}(\mathbf x)\frac {\partial h^{\textup{col}}_{ij}}{\partial x_i}\bigg) \dot {x}_i\Bigg] &\geq -\alpha(\phi(\mathbf x,\mathbf w)) \notag,
}
$\forall \mathbf x \in \Zcal(\mathbf w)$. If such $\alpha(\cdot)$ exists, $\frac {\partial \phi}{\partial \mathbf x}(\mathbf x)\neq 0$ $\forall \mathbf x \in \partial \Zcal(\mathbf w)$, and $\mathbf x(t_0) \in \Zcal(\mathbf w)$, then any Lipschitz continuous controller that satisfies~\eqref{eq:cbf_constraint} would make all robots maintain at least $F'$ neighbors without collisions $\forall t\geq t_0$.
The issue with~\eqref{eq:cbf_constraint} is that it is centralized; it depends on the states and dynamics of all agents. A robot $i \in \Vcal$ cannot compute its control input $u_i \in \Ucal_i$ while ensuring the satisfaction of~\eqref{eq:cbf_constraint} without considering the states and control inputs of all other robots. Various approaches exist to decouple CBF constraints into multiple distributed CBF constraints \cite{jankovic_collision_2024, tan2022, bhatia_decentralized_2024}. We adopt the strategy from~\cite{jankovic_collision_2024, bhatia_decentralized_2024} by dividing the responsibility of satisfying~\eqref{eq:cbf_constraint} among robots.

\subsection{Constraint with Local Information}
In this section, we decompose the CBF constraint~\eqref{eq:cbf_constraint} into $n$ separate \textit{distributed} constraints that only require local information. We first examine the terms $\frac{\partial h_k}{\partial x_i}$, $\frac{\partial h^{\textup{col}}_{ij}}{\partial x_i}$, $E_k(\mathbf x)$, and $E^c_{ij}(\mathbf x)$ in~\eqref{eq:cbf_constraint}. Since robot $i\in \Vcal$ can measure physical states of its own and its neighbors, it can locally compute $\frac{\partial h_k}{\partial x_i}$ $\forall k \in \mathcal{B}_i(t)$ as well as $\frac{\partial h^{\textup{col}}_{ij}}{\partial x_i}$ and $E^c_{ij}(\mathbf x)$ $\forall j \in \Ncal_i(t)$. However, it may not have enough information to determine $E_k(\mathbf x)$ (since $\Ncal_k(t)$ and thus $h_k(\mathbf x)$ may be unknown) for $k\in \Bcal_i(t)$. Thus, each robot constructs
\eqn{\hat h_k(t) = c_k(t) -F', \label{eq:shared_connectivity}}
which may or may not accurately represent $h_k(\mathbf x(t))$~\eqref{eq:h_i} at time $t$, depending on whether robot $k$ updates $c_k(t)$ according to~\eqref{eq:continuous_a}. Let $\hat E_k(t)=e^{-w_r \hat h_k(t)}$. Also, for each $i\in \Vcal$, let
\eqn{\label{eq:phi_i}
\phi_i(\mathbf x, \mathbf w)= \frac {1} {n} - \frac {\sum_{k\in \Bcal_i(t)}\hat E_k(t)}{F'+1} - \frac {\sum_{j \in \Ncal_i(t)} E^c_{ij}(\mathbf x)} {2},
}
which defines a set $\Zcal_i(\mathbf w)=\{\mathbf x \in \R^{nm} \mid \phi_i(\mathbf x, \mathbf w)\geq 0\}$. Due to the structure of~\eqref{eq:phi_i}, $\phi_i(\mathbf x, \mathbf w)\geq0$ at time $t$ only when $\hat {h}_k(t)>0$ $\forall k \in \Bcal_i(t)$ and $h^{\textup{col}}_{ij}(\mathbf x)>0$ $\forall j \in \Ncal_i(t)$ at $t$. Note $\tilde \Zcal(\mathbf w)=\cap_{i\in\Vcal} \Zcal_i(\mathbf w)$ is not necessarily a subset of $\Zcal(\mathbf w)$ as $\hat h_i(t)$ may not be equal to $h_i(\mathbf x(t))$ for some $i\in \Acal$. Now, we present the distributed CBF constraint for agent $i\in \Vcal$:

\eqn{
\label{eq:split}
 \bigg(&\underbrace{w_r\sum_{k\in \Bcal_i(t)} \hat E_k(t) \frac {\partial h_k}{\partial x_i}  + w_c\sum_{j\in \Ncal_i(t)} E^c_{ij}(\mathbf x)\frac {\partial h^{\textup{col}}_{ij}}{\partial x_i}\bigg)}_{H_i(\mathbf x)} \dot {x}_i  \\ & \geq -\alpha(\phi_i(\mathbf x, \mathbf w))\notag,
}where we assume $H_i(\mathbf x)\neq 0, \forall \mathbf x \in \partial \tilde \Zcal_i(\mathbf w)=\{\mathbf x\in \tilde \Zcal(\mathbf w)\mid \phi_i(\mathbf x, \mathbf w)=0\}$, $\forall i \in \Vcal$.

\begin{lemma}
Let $\Gcal(t)=(\Vcal, \Ecal(t))$ be a network consisting of $n$ agents with dynamics \eqref{eq:dynamics} and collective states $\mathbf x\in \tilde \Zcal(\mathbf w)$ for $\mathbf w=[w_r,w_c]\in \R_+^2$ at time $t\geq t_0$. Let $\alpha:\R\to \R$ be a linear extended class $\Kcal_{\infty}$ function. Suppose that in~\eqref{eq:composition} and \eqref{eq:cbf_constraint}, $E_i(\mathbf x(t))$ is replaced by $\hat E_i(t)$. Then,~\eqref{eq:split} $\forall i \in \Vcal$ is feasible. Also, if $u_i(t)$ satisfies~\eqref{eq:split} $\forall i\in \Vcal$ at $t$, $\mathbf u(t)$ satisfies~\eqref{eq:cbf_constraint} at $t$.
\label{lem:equivalence}
\end{lemma}
\begin{proof}
First, we prove the feasibility of~\eqref{eq:split}. Since $\mathbf x \in \tilde \Zcal(\mathbf w)$, we know $\phi_i(\mathbf x, \mathbf w)\geq 0$ $\forall i \in \Vcal$. Then, $u_i(t)=0$ satisfies~\eqref{eq:split} $\forall i \in \Vcal$, regardless of $\alpha(\cdot)$.
Now we prove that satisfying~\eqref{eq:split} $\forall i \in \Vcal$ also implies the satisfaction of~\eqref{eq:cbf_constraint}. Since  $|\Bcal_k(t)|\geq F'+1$ $\forall k\in \Bcal_i(t)$, $\forall i \in \Vcal$, \eqn{\label{eq:proof1}
\sum_{i \in \Vcal}\hat E_i(t)\leq \sum_{i \in \Vcal}\frac {\sum_{k \in \Bcal_i(t)}\hat E_k(t)}{F'+1}.}
Also, because two agents $i,j$ are involved in $h_{ij}^{\text{col}}(\mathbf x)$,
\eqn{\label{eq:proof2}
\sum_{(i,j) \in \Ecal(t)} E^c_{ij}(\mathbf x)= 
\sum_{i \in \Vcal}\frac {\sum_{j \in \Ncal_i(t)} E^c_{ij}(\mathbf x)} {2}.}
Thus, combining~\eqref{eq:proof1} and~\eqref{eq:proof2}, we get
\eqn{\phi(\mathbf x,\mathbf w)\geq \sum_{i\in \Vcal} \phi_i(\mathbf x,\mathbf w).}
Therefore, with any linear function $\alpha(\cdot)$, we can conclude
\seqn{
 \sup_{\mathbf u \in \Ucal} \Bigg[ \sum_{i \in \Vcal} 
\bigg(w_r\sum_{k\in \Bcal_i(t)} \hat E_k(t)\frac {\partial h_k}{\partial x_i} \\  + w_c\sum_{j\in \Ncal_i(t)} E^c_{ij}(\mathbf x)\frac {\partial h^{\textup{col}}_{ij}}{\partial x_i}\bigg) \dot {x}_i\Bigg] & = \notag \\
 \sum_{i\in \Vcal}\sup_{u_i \in \Ucal_i} \Bigg[ \bigg(w_r\sum_{k\in \Bcal_i(t)} 
 \hat E_k(t)\frac {\partial h_k}{\partial x_i} \\ + w_c\sum_{j\in \Ncal_i(t)} E^c_{ij}(\mathbf x)\frac {\partial h^{\textup{col}}_{ij}}{\partial x_i}\bigg)\dot{x}_i \Bigg] & \geq -\sum_{i\in \Vcal}\alpha(\phi_i(\mathbf x, \mathbf w)) \notag\\ &\geq  -\alpha(\phi(\mathbf x, \mathbf w)),
\label{seq:proof}
}
which completes the proof.
\end{proof}

\Cref{lem:equivalence} shows that each robot $i\in \Vcal$ satisfying its own constraint~\eqref{eq:split} leads to all robots collectively satisfying~\eqref{eq:cbf_constraint} with $h_k(\mathbf x(t))$ being replaced by $\hat h_k(t)$, \textit{regardless of whether $\hat h_k(t)=h_k(\mathbf x(t))$ or not}. We also highlight that all the components of~\eqref{eq:split} for robot $i\in \Vcal$ can be derived only using its local information. Using this fact, we now design our distributed controller.

\subsection{Controller Design}
Now, we present our controller. Let $u_{i,\text{des}}(x_i)$ be the desired controller of robot $i$. Our CBF-based Quadratic Program (CBF-QP) controller for robot $i\in \Vcal$ is:
\seqnone{\label{seq:qp}    u^*_{i}(\mathbf x)= &\arg \min\limits_{u_i\in \Ucal_i}||u_{i,\text{des}}(x_i)-u_i||^2 
    \\
     \text{s.t. } 
     & \eqref{eq:split} \text{ is satisfied.}}

 \begin{assumption}
    $u^*_i(\mathbf x)$~\eqref{seq:qp} is locally Lipschitz continuous with respect to $\mathbf x$.
    \label{assum:local_lip}
\end{assumption}

 \begin{remark}
   In general, CBF-QPs with input constraints do not yield locally Lipschitz continuous solutions, as the input constraints effectively act as additional constraints. The local Lipschitz continuity of QPs with multiple constraints have been studied in~\cite{still2018lectures, mestres_regularity_2025, agrawal2025reformulations}.

\end{remark}

Now, we prove that our controller ensures both resilient consensus and physical safety under a specific class of malicious attacks on $c_i(t)$:

\begin{theorem}
    Let all conditions in~\Cref{lem:connectivity}-\ref{lem:equivalence} and~\Cref{assum:attack}-\ref{assum:local_lip} hold. Suppose all misbehaving agents $i\in \Acal$ update $c_i(t)$ according to $c_i(t) = \sum_{j\in \Ncal_i(t)}a_{ij}(\mathbf x(t)) +\epsilon_i$, $\forall t\geq t_0$, where $\epsilon_i\in \R$, such that $\mathbf x(t_0)\in \tilde \Zcal(\mathbf w)$. If each robot $i\in \Vcal$ computes its control input $u_i$ using~\eqref{seq:qp}, then $\forall t \geq t_0$, all robots move such that (a) $\delta_{\Lcal}(\Gcal(t))\geq F'$, (b) no collisions occur, and (c) resilient consensus is guaranteed.
    \label{thm:cbf-qp}
\end{theorem}
\begin{proof}
\textbf{(Part 1)} By~\Cref{assum:attack}, all robots compute $u_i(t)$ from~\eqref{seq:qp} $\forall t \geq t_0$, which means $\mathbf u(t)$ from~\eqref{seq:qp} $\forall i \in \Vcal$ satisfies~\eqref{eq:cbf_constraint} $\forall t \geq t_0$ by~\Cref{lem:equivalence}. Since (a) $u^*_i(\mathbf x)$ $\forall i \in \Vcal$ are locally Lipschitz by~\Cref{assum:local_lip} and feasible $\forall \mathbf x\in \tilde \Zcal(\mathbf w)$ by~\Cref{lem:equivalence}, (b) $\mathbf x(t_0) \in \tilde \Zcal(\mathbf w)$, and (c) $H_i(\mathbf x)\neq 0$ $\forall \mathbf x \in \partial \tilde \Zcal_i(\mathbf w)$, the controllers $u^*_i(\mathbf x)$ $\forall i \in \Vcal$ together render $\tilde \Zcal(\mathbf w)$ forward invariant. 

\textbf{(Part 2)}
Because $\tilde \Zcal(\mathbf w)$ is rendered forward invariant, $\forall t \geq t_0$ (i) $h_i(\mathbf x(t))=\hat h_i(t)> 0$ $\forall i\in \Lcal$, (ii) $h_i(\mathbf x(t))=\hat h_i(t)-\epsilon_i> 0$ $\forall i\in \Acal$, and (iii) $h^{\textup{col}}_{ij}(\mathbf x(t))>0$ $\forall i\in \Vcal$ $\forall j \in \Ncal_i(t)$. Then, by~\Cref{lem:connectivity}, robots are guaranteed to move such that $\delta_{\Lcal}(\Gcal(t))\geq F'$ without collision $\forall t \geq t_0$. Also, by~\Cref{thm:rs_robust}, resilient consensus is guaranteed. 
\end{proof}

\begin{remark}
Our analysis in~\Cref{thm:cbf-qp} assumes malicious robots conduct constant biased attacks on $c_i(t)$ (but not necessarily on $y_i(t)$), which are common forms of attacks due to their stealthiness~\cite{bias_attack1}. However, these represent a specific class of malicious attacks. Extending our analysis to broader attack strategies is future work.
\end{remark}

We also examine the performance of our controller~\eqref{seq:qp} in the cases where all robots $i \in \Vcal$ correctly update connectivity levels $c_i(t)$ according to the nominal update rule~\eqref{eq:continuous_a} below:
\begin{corollary}
Let all conditions in~\Cref{thm:cbf-qp} hold. Suppose all robots $i\in \Vcal$ update $c_i(t)$ according to~\eqref{eq:continuous_a}, i.e., $\epsilon_i=0$, $\forall i \in \Acal$, such that $\mathbf x(t_0)\in \tilde \Zcal(\mathbf w)$. If each robot $i\in \Vcal$ computes its control input $u_i$ using~\eqref{seq:qp}, then $\forall t\geq t_0$, all robots move such that (a) $\delta_{\min}(\Gcal(t))\geq F'$, (b) no collisions occur, and (c) resilient consensus is guaranteed.
\label{cor:naive}
\end{corollary}
\begin{proof}
\textbf{(Part 1)} By the same argument as in (Part 1) of the proof of~\Cref{thm:cbf-qp}, the controllers $u^*_i(\mathbf x)$ $\forall i \in \Vcal$ together render $\tilde \Zcal(\mathbf w)$ forward invariant.

\textbf{(Part 2)} Because $\epsilon_i=0$ $\forall i \in \Acal$ and $\tilde \Zcal(\mathbf w)$ is rendered forward invariant, $\forall t \geq t_0$ (i) $h_i(\mathbf x(t))=\hat h_i(t)> 0$ $\forall i\in \Vcal$ and (ii) $h^{\textup{col}}_{ij}(\mathbf x(t))>0$ $\forall i\in \Vcal$ $\forall j \in \Ncal_i(t)$. Then, by~\Cref{lem:connectivity}, robots are guaranteed to move such that $\delta_{\min}(\Gcal(t))\geq F'$ without collision $\forall t \geq t_0$. Also, by~\Cref{thm:rs_robust}, resilient consensus is guaranteed. 
\end{proof}
\begin{remark}
\label{remark:diff}
As shown in~\Cref{thm:cbf-qp} and~\Cref{cor:naive}, while resilient consensus and collision avoidance are always ensured, the performance of our designed controller~\eqref{seq:qp} (i.e., the overall connectivity of the network) depends on the behavior of malicious robots. If all robots $i\in \Vcal$ update their connectivity values $c_i(t)$ according to~\eqref{eq:continuous_a}, then \textit{all robots} are guaranteed to maintain at least $F'$ neighbors while avoiding collisions. However, if malicious robots $i \in \Acal$ manipulate ${c}_i(t)$, they may fail to maintain $F'$ neighbors. This also affects how aggressively/conservatively robots can move with our controller. We examine this effect more closely in our simulation results.
\end{remark}

\section{Simulation Results}
\label{sec:sim}
\begin{figure*}[ht!]
    \centering
\includegraphics[width=0.95\linewidth]{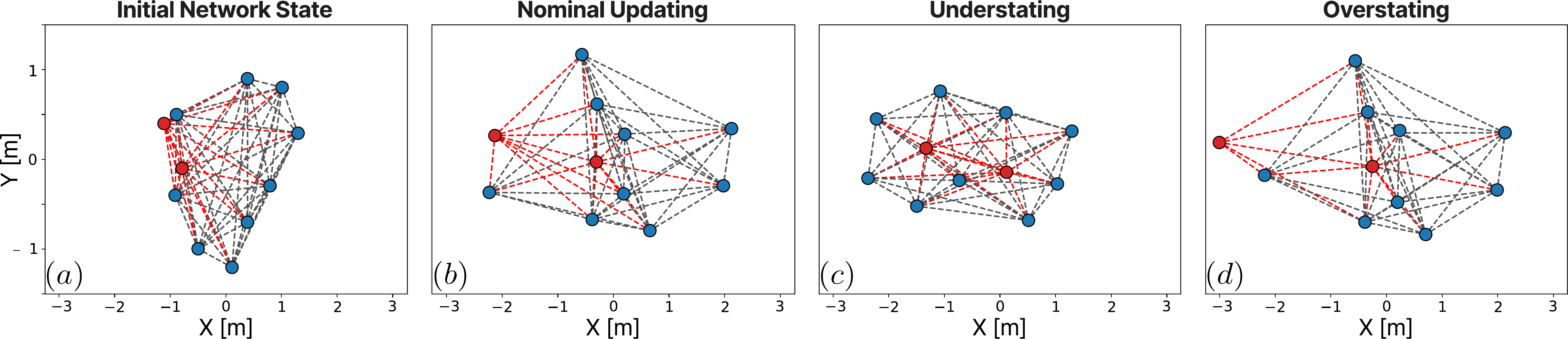}
    \caption{(a) shows the initial robot formation at $t_0 = 0$, while (b)–(d) depict formations at $t = 10$ seconds for nominal updating, understating, and overstating scenarios, respectively. Gray dotted lines are connections between normal robots, while red dotted lines are those involving malicious robots.}
    \label{fig:sim}
\end{figure*}

\begin{figure}
    \centering
    \includegraphics[width=0.97\linewidth]{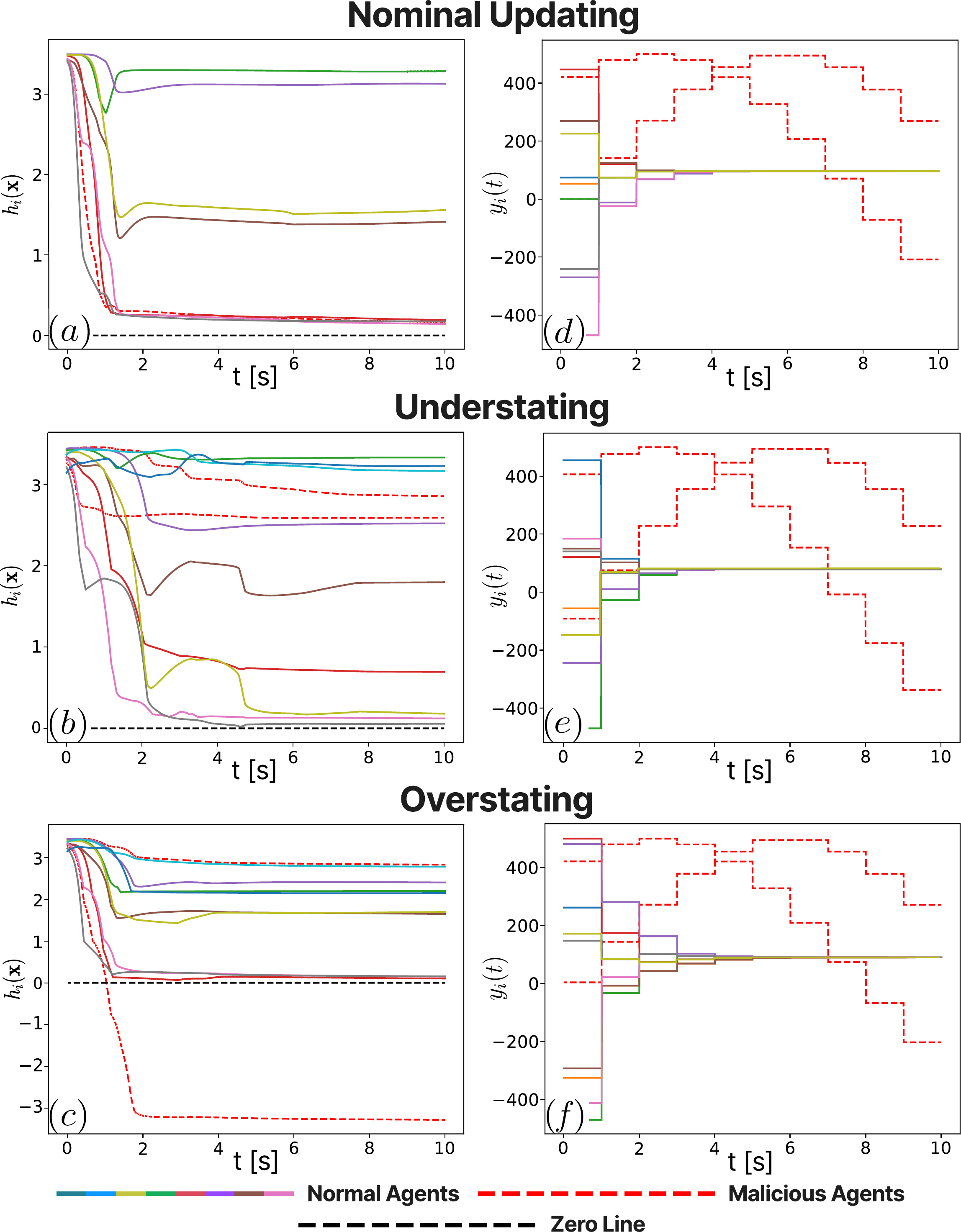}
    \caption{(a)–(c) show the robots' actual values of~\eqref{eq:h_i}, while (d)–(f) present the robots' consensus states, both for the simulations of nominal updating, understating, and overstating scenarios for $t\in[0,10]$, respectively.}
    \label{fig:sim_result}
\end{figure}
Here, we show the effectiveness of our controller~\eqref{seq:qp} through simulations in a $2$D plane. We consider a time-varying network $\Gcal(t)=(\Vcal,\Ecal(t))$, consisting of $n=11$ robots, where edges are formed based on robots' relative distances with a communication range $R=3$ m. Each robot has dynamics~\eqref{eq:dynamics} and an input set $\Ucal_i=[-1.5,15]\times [-1.5,1.5]$. They are simulated as point masses and to maintain a minimum distance of $\Delta_d = 0.3$ m. Each normal robot $i\in \Lcal$ updates its connectivity value $c_i(t)$ according to~\eqref{eq:connectivity} and shares it with its neighbors every $\tau_1 = 0.005$ seconds. It also starts with a random consensus state $y_i(t)\in[-500,500]$ at $t=t_0$, which it shares and updates every $\tau_2 = 0.1$ seconds using W-MSR algorithm thereafter. We examine three scenarios, each involving $F=2$ malicious robots $i\in \Acal$ following different update rules for $c_i(t)$:

\noindent\textbf{1. Nominal Updating:} Follow the nominal law~\eqref{eq:connectivity} to update: $$c_i(t) = \sum_{j\in \Ncal_i(t_l)}a_{ij}(\mathbf x(t_l)), \ \forall t \in [t_l, t_{l+1}),\  \forall l\in \Z_{\geq 0}.$$\noindent\textbf{2. Understating: }Understate their connectivity levels: $$c_i(t) = \sum_{j\in \Ncal_i(t_l)}a_{ij}(\mathbf x(t_l))-2.5, \ \forall t \in [t_l, t_{l+1}), \ \forall l\in \Z_{\geq 0}.$$\noindent\textbf{3. Overstating: }Overstate their connectivity levels:$$c_i(t) = \sum_{j\in \Ncal_i(t_l)}a_{ij}(\mathbf x(t_l))+3.5, \ \forall t \in [t_l, t_{l+1}), \ \forall l\in \Z_{\geq 0}.$$ In all cases, all malicious robots randomly choose and share $y_i(t)\in [-500,500]$ every $\tau_2$ seconds. By~\Cref{thm:rs_robust}, all normal robots have to maintain $F'=2+\big\lfloor \frac {11}{2} \big\rfloor= 7$ or more neighbors to ensure resilient consensus. We evaluate our controller by observing whether robots can achieve resilient consensus even when the robots' desired control inputs directly conflict with maintaining network connectivity, despite malicious robots' adversarial behaviors. With the desired controller $u_{i,\text{des}}(x_i) = \frac {v(x_i)}{\norm{v(x_i)}}$ for each robot $i\in \Vcal$, where
\eqnN{v(x_i) = \begin{cases}
    (-1)^i\cdot[100, 0]^T - x_i & \text{ if $i \in \{1,\dots,5\} $},\\
     (-1)^i\cdot[0, 100]^T -x _i & \text{ if $i \in \{6,\dots, 11\} $},
\end{cases}}
robots will move in four different directions, which would break the network and prevent resilient consensus. Each simulation begins with every robot having at least seven neighbors, as shown in Fig.~\ref{fig:sim}~(a), where normal and malicious robots are blue and red, respectively. Links between normal robots are depicted as gray dotted lines, while links involving malicious robots are shown as red dotted lines.

The robot formations at $t=10$ seconds for each scenario are shown in Fig.~\ref{fig:sim}~(b)-(d). Regardless of whether malicious robots $i\in \Acal$ update $c_i(t)$ according to~\eqref{eq:connectivity}, robots eventually stop dispersing to allow normal robots maintain at least seven neighbors while avoiding collisions, validating~\Cref{thm:cbf-qp} and~\Cref{cor:naive}. Fig.~\ref{fig:sim_result}~(a)-(c) show evolution of the degree constraint functions~\eqref{eq:h_i} for each scenario. The results also highlight the effect of malicious behaviors on robots' movement, which is discussed in~\Cref{remark:diff}. Comparing Fig.~\ref{fig:sim}~(b) and~(c), we observe that understating behaviors impose more connectivity constraints for normal robots to remain close to malicious agents, leading to a more compact network. Conversely, as seen in Fig.~\ref{fig:sim}~(b) and (d), overstating behaviors relax these constraints, allowing normal agents to move more freely and form a more dispersed network. Notably, one malicious robot fails to maintain seven neighbors (shown in Fig.~\ref{fig:sim_result}~(c)). Importantly, these variations in updates do not compromise resilient consensus (as shown in Fig.~\ref{fig:sim_result}~(d)-(f)), demonstrating the effectiveness of our controller. 

\section{Conclusion}

\label{sec:conc}
In this work, we propose a distributed CBF-based controller that guarantees resilient consensus among robots during navigation without relying on fixed network topologies. We first establish a sufficient condition for resilient consensus based solely on the minimum degree of normal agents in a network. Leveraging this condition, we design a controller that enables robots to make movement decisions using only local information, ensuring resilient consensus in a time-varying network. Furthermore, our current analysis assumes continuous connectivity information and a specific form of malicious attack, while real systems operate with discrete updates and potentially varied threats. We aim to address these issues in future extensions.






\bibliographystyle{IEEEtran}
\bibliography{references_ll}

\end{document}